\documentclass{article}
\usepackage{amssymb}

\usepackage{amsmath}

\newtheorem{theorem}{Theorem}[section]

\newtheorem{conjecture}[theorem]{Conjecture}

\newtheorem{definition}[theorem]{Definition}

\newtheorem{lemma}[theorem]{Lemma}

\newtheorem{Main Result:}{Main Result:}
\newenvironment{proof}[1][Proof]{\textbf{#1.} }{\ \rule{0.5em}{0.5em}} %
\begin{document}

\title{{Proof of a conjecture about rotation symmetric functions}
\thanks{This work was supported by NSF of China with contract No.
60803154}
\author{Zhang Xiyong$^{1}$\footnote{\textbf{Corresponding E-mail Address:} xiyong.zhang@hotmail.com}, \ \
Guo Hua$^{2}$,\ \  Li Yifa$^{1}$ \\
{\small 1.Zhengzhou Information Science and Technology Institute, PO Box 1001-745, Zhengzhou 450002, PRC} \\
{\small 2.School of Computer Science and Engineering, Beihang
University, Beijing,100083, PRC.} }
\date{}}
\maketitle

\vspace{-0.4cm}

\begin{abstract}
Rotation symmetric Boolean functions have important applications
in the design of cryptographic algorithms. We prove the conjecture
about rotation symmetric Boolean functions (RSBFs) of degree 3
proposed in \cite{cusick}, thus the nonlinearity of such kind of
functions are determined.
\end{abstract}

\par \textbf{Keywords:}
{\textit{Boolean functions, Rotation-symmetric, Fourier Transform,
Nonlinearity}

\section{Introduction}
A Boolean function $f^n(x_0,\cdots,x_{n-1})$ on $n$ variables is a
map from $\mathbb{F}_2^n$ to $\mathbb{F}_2$, where
$\mathbb{F}_2^n$ is the vector space of dimension $n$ over the two
element field $\mathbb{F}_2$. Rotation symmetric Boolean functions
(Abbr. RSBFs) are a special kind of Boolean functions with
properties that its evaluations on every cyclic inputs are the
same, thus could be used as components to achieve efficient
implementation in the design of a message digest algorithm in
cryptography, such as MD4, MD5. These functions have attracted
attentions in these years (see [2-7]). One of the main focus is
the nonlinearity of these kind functions (see \cite{kavut,kim}).
It is known that a hashing algorithm employing degree-two RSBFs as
components cannot resist the linear and differential attacks
(\cite{marial}). Hence, it is necessary to use higher degree RSBFs
with higher nonlinearity to protect the cryptography algorithm
from differential attack. Cusick and St\u{a}nic\u{a}
(\cite{cusick}) investigated the weight of a kind of 3-degree
RSBFs and proposed a conjecture based on their numerical
observations.

\begin{conjecture}
The nonlinearity of $F_3^n(x_0,\cdots,x_{n-1})=\sum\limits_{0\leq
i\leq n-1}x_ix_{i+1(mod\ n)}x_{i+2(mod\ n)}$ is the same as its
weight.
\end{conjecture}

As claimed in \cite{cusick} that if the above Conjecture could be
proved, then significant progress for $k$-degree ($k>3$) RSBFs
might be possible. Recently Ciungu \cite{ciungu} proved the
conjecture in the case $3|n$. In this paper, we factor $F_3^n$
into four sub-functions, discover some recurrence relations, and
thus prove the above Conjecture. The sub-functions and related
recurrence are different from Cusick's\cite{cusick}. The technique
used in this paper may be applied for the study of RSBFs of degree
$k>3$.

We define two vectors $e_1=(1,0,\cdots,0)\in \mathbb{F}_2^n$ for
every $n>1$, $e_{2^{n-1}}=(0,0,\cdots,0,1)\in \mathbb{F}_2^n$, and
abuse $0=(0,\cdots,0)$ to represent the zero vector in vector
spaces $\mathbb{F}_2^n$ of every dimension for simpleness. By
$x^n$ and $c^n$ we mean the abbr. forms of vectors
$(x_0,\cdots,x_{n-1})$ and $(c_0,\cdots,c_{n-1})$ in
$\mathbb{F}_2^n$. A linear function is of the form $c^n\cdot x^n$,
where $\cdot$ is the vector dot product. The \textit{weight} of a
Boolean function $f^n(x^n)$ is the number of solutions $x^n\in
\mathbb{F}_2^n$ such that $f^n(x^n)=1$, denoted by $wt(f^n)$. The
distance $d(f^n,g^n)$ between two Boolean functions $f^n$ and
$g^n$ is defined to be $wt(f^n+g^n)$.

Now we list some basic definitions about Boolean functions.

\begin{definition}\label{rots}
A Boolean function $f^n(x^n)$, is called rotation symmetric if
$$f^n(x_0,\cdots,x_{n-1})=f^n(x_{n-1},x_0,x_1,\cdots,x_{n-2}),\ for\ all\ (x_0,\cdots,x_{n-1})\in \mathbb{F}_2^n.$$
\end{definition}

\begin{definition}\label{fourier}
For a Boolean function $f^n(x^n)$, the Fourier transform of $f^n$
at $c^n\in  \mathbb{F}_2^n$ is defined as
$$\widehat{f^n}(c^n)=\sum\limits_{x^n\in \mathbb{F}_2^n}(-1)^{f^n(x^n)+c^n\cdot x^n}.$$
\end{definition}

By the definition of Fourier transform, it is easy to see that
\begin{lemma} For all $(c_0,\cdots,c_{n-1})\in \mathbb{F}_2^n$,
$$\widehat{F_3^n}(c_0,\cdots,c_{n-1})=\widehat{F_3^n}(c_{n-1},c_0,\cdots,c_{n-2}).$$
\end{lemma}

\begin{definition}\label{nonlinearity}
The nonlinearity $N_{f^n}$ of a Boolean function $f^n(x^n)$, is
defined as
$$N_{f^n}=Min\ \{d(f^n(x^n),c^n\cdot x^n)|c^n\in  \mathbb{F}_2^n\}.$$
\end{definition}

By Definition \ref{nonlinearity}, it is not difficult to deduce
that for all $f^n(x^n)$,
$$\widehat{f^n}(0)=2^n-2\cdot wt(f^n(x^n)).$$
\noindent Hence we can restate the above Conjecture as
$$\widehat{F_3^n}(0)=Max\{ |\widehat{F_3^n}(c^n)| \ | c^n\in  \mathbb{F}_2^n\}.$$

\section{The proof of the Conjecture}
To prove the above Conjecture, we factor $F_3^n$ into 4
sub-functions. Let  $t_{n}=\sum\limits_{0\leq i\leq
n-3}x_ix_{i+1}x_{i+2}$, and
\begin{equation}
\begin{array}{ll}
&f_0^{n}(x_0,\cdots,x_{n-1})=t_n,\\
&f_1^{n}(x_0,\cdots,x_{n-1})=t_n+x_0x_1,\\
&f_2^{n}(x_0,\cdots,x_{n-1})=t_n+x_{n-2}x_{n-1},\\
&f_3^{n}(x_0,\cdots,x_{n-1})=t_n+x_0x_1+x_{n-2}x_{n-1}+x_0+x_{n-1}.
\end{array}
\end{equation}
Then we have
$$\sum\limits_{x_0,\cdots,x_{n-1}}(-1)^{F_3^n(x_0,\cdots,x_{n-1})}=\sum\limits_{x_0,\cdots,x_{n-3}}\sum\limits_{0\leq i\leq 3}(-1)^{f_i^{n-2}(x_0,\cdots,x_{n-3})}.$$

\begin{lemma}\label{high0}
For every $c^n=(c_0,\cdots,c_{n-1})\in \mathbb{F}_2^n$, if
$c_{n-1}=0$, then
\begin{equation}
\begin{array}{ll}
&\widehat{f_0^n}(c^n)=2(\widehat{f_0^{n-2}}(c^{n-2})+(-1)^{c_{n-2}}\cdot\widehat{f_0^{n-3}}(c^{n-3})),\\

&\widehat{f_1^n}(c^n)=2(\widehat{f_1^{n-2}}(c^{n-2})+(-1)^{c_{n-2}}\cdot\widehat{f_1^{n-3}}(c^{n-3})),\\

&\widehat{f_2^n}(c^n)=2(\widehat{f_0^{n-2}}(c^{n-2})+(-1)^{c_{n-3}+c_{n-2}}\cdot\widehat{f_2^{n-3}}(c^{n-3}+e_{2^{n-4}})),\\

&\widehat{f_3^n}(c^n)=2(-1)^{c_{n-2}}\cdot
\widehat{f_1^{n-3}}(c^{n-3}+e_1),

\end{array}
\end{equation}
where $c^{n-2}\in \mathbb{F}_2^{n-2}$ and $c^{n-3}\in
\mathbb{F}_2^{n-3}$ are the first $n-2$ and $n-3$ bits of $c^n\in
\mathbb{F}_2^{n}$, and
$e_1=(1,0,\cdots,0),e_{2^{n-4}}=(0,\cdots,0,1)\in
\mathbb{F}_2^{n-3}$.
\end{lemma}

\begin{proof}
We prove the first relation,  proof of the other three  ones are
similar. Because $c_{n-1}=0$, we have

\begin{equation*}
\begin{array}{ll}
&\widehat{f_0^n}(c^n)\\

&=\sum\limits_{x^n:x_{n-1}=0}(-1)^{f_0^n(x^n)+c^n\cdot
x^n}+\sum\limits_{x^n:x_{n-1}=1}(-1)^{f_0^n(x^n)+c^n\cdot x^n}\\

&=\sum\limits_{x^{n-1}}(-1)^{f_0^{n-1}(x^{n-1})+c^{n-1}\cdot
x^{n-1}}+\sum\limits_{x^{n-1}}(-1)^{f_0^{n-1}(x^{n-1})+x_{n-3}x_{n-2}+c^{n-1}\cdot x^{n-1}}\\

&=\sum\limits_{x^{n-1}:x_{n-2}=0}(-1)^{f_0^{n-1}(x^{n-1})+c^{n-1}\cdot
x^{n-1}}+\sum\limits_{x^{n-1}:x_{n-2}=0}(-1)^{f_0^{n-1}(x^{n-1})+x_{n-3}x_{n-2}+c^{n-1}\cdot x^{n-1}}\\
&\ \ \ \ \
+\sum\limits_{x^{n-1}:x_{n-2}=1}(-1)^{f_0^{n-1}(x^{n-1})+c^{n-1}\cdot
x^{n-1}}+\sum\limits_{x^{n-1}:x_{n-2}=1}(-1)^{f_0^{n-1}(x^{n-1})+x_{n-3}x_{n-2}+c^{n-1}\cdot x^{n-1}}\\
\end{array}
\end{equation*}

\begin{equation*}
\begin{array}{ll}

&=\sum\limits_{x^{n-2}}(-1)^{f_0^{n-2}(x^{n-2})+c^{n-2}\cdot
x^{n-2}}+\sum\limits_{x^{n-2}}(-1)^{f_0^{n-2}(x^{n-2})+c^{n-2}\cdot x^{n-2}}\\
&\ \ \ \ \
+\sum\limits_{x^{n-2}}(-1)^{f_0^{n-2}(x^{n-2})+c^{n-2}\cdot
x^{n-2}+x_{n-4}x_{n-3}+c_{n-2}}\\
&\ \ \ \ \
+\sum\limits_{x^{n-2}}(-1)^{f_0^{n-2}(x^{n-2})+c^{n-2}\cdot x^{n-2}+x_{n-4}x_{n-3}+x_{n-3}+c_{n-2}}\\

&=2\cdot \widehat{f_0^{n-2}}(c^{n-2})\\
&\ \ \ \ \
+\sum\limits_{x^{n-2}:x_{n-3}=0}(-1)^{f_0^{n-2}(x^{n-2})+c^{n-2}\cdot
x^{n-2}+x_{n-4}x_{n-3}+c_{n-2}}\\
&\ \ \ \ \
+\sum\limits_{x^{n-2}:x_{n-3}=1}(-1)^{f_0^{n-2}(x^{n-2})+c^{n-2}\cdot
x^{n-2}+x_{n-4}x_{n-3}+c_{n-2}}\\
&\ \ \ \ \
+\sum\limits_{x^{n-2}:x_{n-3}=0}(-1)^{f_0^{n-2}(x^{n-2})+c^{n-2}\cdot x^{n-2}+x_{n-4}x_{n-3}+x_{n-3}+c_{n-2}}\\
&\ \ \ \ \
+\sum\limits_{x^{n-2}:x_{n-3}=1}(-1)^{f_0^{n-2}(x^{n-2})+c^{n-2}\cdot
x^{n-2}+x_{n-4}x_{n-3}+x_{n-3}+c_{n-2}}\\

&=2\cdot \widehat{f_0^{n-2}}(c^{n-2})\\
&\ \ \ \ \
+\sum\limits_{x^{n-3}}(-1)^{f_0^{n-3}(x^{n-3})+c^{n-3}\cdot
x^{n-3}+c_{n-2}}\\
&\ \ \ \ \
+\sum\limits_{x^{n-3}}(-1)^{f_0^{n-3}(x^{n-3})+c^{n-3}\cdot x^{n-3}+c_{n-2}}\\
&\ \ \ \ \
+\sum\limits_{x^{n-3}}(-1)^{f_0^{n-3}(x^{n-3})+c^{n-3}\cdot
x^{n-3}+x_{n-5}x_{n-4}+x_{n-4}+c_{n-3}+c_{n-2}}\\
&\ \ \ \ \
+\sum\limits_{x^{n-3}}(-1)^{f_0^{n-3}(x^{n-3})+c^{n-3}\cdot
x^{n-3}+x_{n-5}x_{n-4}+x_{n-4}+c_{n-3}+c_{n-2}+1}\\

&=2\cdot \widehat{f_0^{n-2}}(c^{n-2})+2\cdot (-1)^{c_{n-2}}\cdot \widehat{f_0^{n-3}}(c^{n-3})\\
&\ \ \ \ \
+\sum\limits_{x^{n-3}}(-1)^{f_0^{n-3}(x^{n-3})+c^{n-3}\cdot
x^{n-3}+x_{n-5}x_{n-4}+x_{n-4}+c_{n-3}+c_{n-2}}\\
&\ \ \ \ \
-\sum\limits_{x^{n-3}}(-1)^{f_0^{n-3}(x^{n-3})+c^{n-3}\cdot
x^{n-3}+x_{n-5}x_{n-4}+x_{n-4}+c_{n-3}+c_{n-2}}\\

&=2\cdot \widehat{f_0^{n-2}}(c^{n-2})+2\cdot (-1)^{c_{n-2}}\cdot
\widehat{f_0^{n-3}}(c^{n-3}).
\end{array}
\end{equation*}

\end{proof}

\begin{lemma}\label{high1}
For every $c^n=(c_0,\cdots,c_{n-1})\in \mathbb{F}_2^n$, if
$c_{n-1}=1$, then for $i=0, 2$,
\begin{equation}
\begin{array}{ll}
\widehat{f_i^n}(c^n)
&=\widehat{f_0^{n-1}}(c^{n-1})\pm2\cdot\widehat{f_0^{n-4}}(c^{n-4})), \\
\ \ \
or&=\widehat{f_0^{n-1}}(c^{n-1})\pm2\cdot\widehat{f_0^{n-4}}(c^{n-4})\pm4\cdot\widehat{f_2^{n-5}}(c^{n-5}),
 \\
\ \ \
or&=\widehat{f_0^{n-1}}(c^{n-1})\pm2\cdot\widehat{f_0^{n-4}}(c^{n-4})\pm4\cdot\widehat{f_2^{n-5}}(c^{n-5}+e_{2^{n-6}}),

\end{array}
\end{equation}

and for $i=1$,
\begin{equation}
\begin{array}{ll}
\widehat{f_i^n}(c^n)
&=\widehat{f_1^{n-1}}(c^{n-1})\pm2\cdot\widehat{f_1^{n-4}}(c^{n-4}),\\
\ \ \
or&=\widehat{f_1^{n-1}}(c^{n-1})\pm2\cdot\widehat{f_1^{n-4}}(c^{n-4})\pm4\cdot\widehat{f_1^{n-5}}(c^{n-5}),
\\
\ \ \
or&=\widehat{f_1^{n-1}}(c^{n-1})\pm2\cdot\widehat{f_1^{n-4}}(c^{n-4})\pm4\cdot\widehat{f_3^{n-5}}(c^{n-5}+e_{1}),
\end{array}
\end{equation}

and for $i=3$,
\begin{equation}
\begin{array}{ll}
\widehat{f_i^n}(c^n)
&=\widehat{f_1^{n-1}}(c^{n-1}+e_{1})\pm2\cdot\widehat{f_1^{n-4}}(c^{n-4}+e_{1}),\\
\ \ \
or&=\widehat{f_1^{n-1}}(c^{n-1}+e_{1})\pm2\cdot\widehat{f_1^{n-4}}(c^{n-4}+e_{1})\pm4\cdot\widehat{f_1^{n-5}}(c^{n-5}+e_{1}),
\\
\ \ \
or&=\widehat{f_1^{n-1}}(c^{n-1}+e_{1})\pm2\cdot\widehat{f_1^{n-4}}(c^{n-4}+e_{1})\pm4\cdot\widehat{f_3^{n-5}}(c^{n-5}),

\end{array}
\end{equation}

where $c^{n-1}\in \mathbb{F}_2^{n-1}$, $c^{n-4}\in
\mathbb{F}_2^{n-4}$, and $c^{n-5}\in \mathbb{F}_2^{n-5}$ are the
first $n-1$ , $n-4$ and $n-5$ bits of $c^n\in \mathbb{F}_2^{n}$,
and $e_1=(1,0,\cdots,0),e_{2^{n-6}}=(0,\cdots,0,1)\in
\mathbb{F}_2^{n-5}$.
\end{lemma}

\begin{proof}
We briefly prove the relations for $f_0^n,f_2^n$.

Because $c_{n-1}=1$, we have
\begin{equation}
\begin{array}{ll}
&\widehat{f_0^n}(c^n)\\

&=\sum\limits_{x^n:x_{n-1}=0}(-1)^{f_0^n(x^n)+c^n\cdot
x^n}+\sum\limits_{x^n:x_{n-1}=1}(-1)^{f_0^n(x^n)+c^n\cdot x^n}\\

&=\widehat{f_0^{n-1}}(c^{n-1})+\sum\limits_{0\leq j\leq
7}(-1)^{g_{0,j}^{n-4}}.
\end{array}
\end{equation}
where $g_{0,j}^{n-4}(x_0,\cdots,x_{n-5})$ are functions
corresponding to $f_0^n(x^n)+c^n\cdot x^n$ where
$c_{n-1}=1,x_{n-1}=1, j=x_{n-4}+2x_{n-3}+4x_{n-2}$. Calculate
these functions in details in Table \ref{example-g1}.
\begin{table}[thb]
\begin{center}
\begin{tabular}{|c|c|}\hline
  $j:(x_{n-4},x_{n-3},x_{n-2})$ & $g_{0,j}^{n-4}$ \\\hline
  $(0,0,0)$ & $f_0^{n-4}+c^{n-4}\cdot x^{n-4}+1$ \\\hline
  $(1,0,0)$ & $f_0^{n-4}+c^{n-4}\cdot x^{n-4}+x_{n-6}x_{n-5}+c_{n-4}+1$ \\\hline
  $(0,1,0)$ & $f_0^{n-4}+c^{n-4}\cdot x^{n-4}+c_{n-3}+1$ \\\hline
  $(0,0,1)$ & $f_0^{n-4}+c^{n-4}\cdot x^{n-4}+c_{n-2}+1$ \\\hline
  $(1,1,0)$ & $f_0^{n-4}+c^{n-4}\cdot x^{n-4}+x_{n-6}x_{n-5}+x_{n-5}+c_{n-4}+c_{n-3}+1$ \\\hline
  $(1,0,1)$ & $f_0^{n-4}+c^{n-4}\cdot x^{n-4}+x_{n-6}x_{n-5}+c_{n-4}+c_{n-2}+1$ \\\hline
  $(0,1,1)$ & $f_0^{n-4}+c^{n-4}\cdot x^{n-4}+c_{n-3}+c_{n-2}$ \\\hline
  $(1,1,1)$ & $f_0^{n-4}+c^{n-4}\cdot x^{n-4}+x_{n-6}x_{n-5}+x_{n-5}+c_{n-4}+c_{n-3}+c_{n-2}+1$ \\\hline
\end{tabular}
\end{center}
\caption{$g_{0,j}^{n-4} (0\leq j\leq 7)$ corresponding to
$f_0^n(x^n)+c^n\cdot x^n$.} \label{example-g1}
\end{table}

By Table \ref{example-g1}, we have
\begin{equation}
\begin{array}{ll}
&\sum\limits_{0\leq j\leq 7}(-1)^{g_{0,j}^{n-4}}\\
&=((-1)+(-1)^{c_{n-2}+1}+(-1)^{c_{n-3}+1}+(-1)^{c_{n-3}+c_{n-2}})\cdot
\widehat{f_0^{n-4}}(c^{n-4}) \\
&\ \ \ \ \ +(-1)^{c_{n-4}+1}(1+(-1)^{c_{n-2}})\cdot
\widehat{f_2^{n-4}}(c^{n-4})\\
&\ \ \ \ \ +(-1)^{c_{n-4}+c_{n-3}+1}(1+(-1)^{c_{n-2}})\cdot
\widehat{f_2^{n-4}}(c^{n-4}+e_{2^{n-5}})\\

&=\left\{
\begin{array}{ll}
-2(-1)^{c_{n-3}}\widehat{f_0^{n-4}}(c^{n-4})\ &\ if\ c_{n-2}=1, \\
-2\widehat{f_0^{n-4}}(c^{n-4})-4(-1)^{c_{n-4}}\widehat{f_0^{n-5}}(c^{n-5})\ &\ if\ c_{n-2}=0, c_{n-3}=0, \\
-2\widehat{f_0^{n-4}}(c^{n-4})-4(-1)^{c_{n-4}+c_{n-5}}\widehat{f_2^{n-5}}(c^{n-5}+e_{2^{n-6}})\ &\ if\ c_{n-2}=0, c_{n-3}=1.  %
\end{array}
\right.

\end{array}
\end{equation}

So we have
\begin{equation}
\begin{array}{ll}
&\widehat{f_0^n}(c^n)\\
&=\left\{
\begin{array}{ll}
\widehat{f_0^{n-1}}(c^{n-1})-2(-1)^{c_{n-3}}\widehat{f_0^{n-4}}(c^{n-4})\ &\ if\ c_{n-2}=1, \\

\widehat{f_0^{n-1}}(c^{n-1})-2\widehat{f_0^{n-4}}(c^{n-4})-4(-1)^{c_{n-4}}\widehat{f_0^{n-5}}(c^{n-5})\ &\ if\ c_{n-2}=0, c_{n-3}=0, \\

\widehat{f_0^{n-1}}(c^{n-1})-2\widehat{f_0^{n-4}}(c^{n-4})-4(-1)^{c_{n-4}+c_{n-5}}\widehat{f_2^{n-5}}(c^{n-5}+e_{2^{n-6}})\
&\ if\ c_{n-2}=0, c_{n-3}=1.
\end{array}
\right.
\end{array}
\end{equation}

For the proof of the relation of $f_2^n$, we list the functions
$g_{2,j}^{n-4} (0\leq j\leq 7)$ corresponding to
$f_2^n(x^n)+c^n\cdot x^n$ in Table \ref{example-g2}, where
$c_{n-1}=1,x_{n-1}=1, j=x_{n-4}+2x_{n-3}+4x_{n-2}$.

\begin{table}[thb]
\begin{center}
\begin{tabular}{|c|c|}\hline
  $j:(x_{n-4},x_{n-3},x_{n-2})$ & $g_{2,j}^{n-4}$ \\\hline
  $(0,0,0)$ & $f_0^{n-4}+c^{n-4}\cdot x^{n-4}+1$ \\\hline
  $(1,0,0)$ & $f_0^{n-4}+c^{n-4}\cdot x^{n-4}+x_{n-6}x_{n-5}+c_{n-4}+1$ \\\hline
  $(0,1,0)$ & $f_0^{n-4}+c^{n-4}\cdot x^{n-4}+c_{n-3}+1$ \\\hline
  $(0,0,1)$ & $f_0^{n-4}+c^{n-4}\cdot x^{n-4}+c_{n-2}$ \\\hline
  $(1,1,0)$ & $f_0^{n-4}+c^{n-4}\cdot x^{n-4}+x_{n-6}x_{n-5}+x_{n-5}+c_{n-4}+c_{n-3}+1$ \\\hline
  $(1,0,1)$ & $f_0^{n-4}+c^{n-4}\cdot x^{n-4}+x_{n-6}x_{n-5}+c_{n-4}+c_{n-2}$ \\\hline
  $(0,1,1)$ & $f_0^{n-4}+c^{n-4}\cdot x^{n-4}+c_{n-3}+c_{n-2}+1$ \\\hline
  $(1,1,1)$ & $f_0^{n-4}+c^{n-4}\cdot x^{n-4}+x_{n-6}x_{n-5}+x_{n-5}+c_{n-4}+c_{n-3}+c_{n-2}$ \\\hline
\end{tabular}
\end{center}
\caption{$g_{2,j}^{n-4} (0\leq j\leq 7)$ corresponding to
$f_2^n(x^n)+c^n\cdot x^n$.} \label{example-g2}
\end{table}

Similarly
\begin{equation}\label{ff22-1}
\begin{array}{ll}
&\widehat{f_2^n}(c^n)\\

&=\sum\limits_{x^n:x_{n-1}=0}(-1)^{f_2^n(x^n)+c^n\cdot
x^n}+\sum\limits_{x^n:x_{n-1}=1}(-1)^{f_2^n(x^n)+c^n\cdot x^n}\\

&=\widehat{f_0^{n-1}}(c^{n-1})+\sum\limits_{0\leq j\leq
7}(-1)^{g_{2,j}^{n-4}},
\end{array}
\end{equation}

And

\begin{equation}\label{ff22-2}
\begin{array}{ll}
&\sum\limits_{0\leq j\leq 7}(-1)^{g_{2,j}^{n-4}}\\
&=((-1)+(-1)^{c_{n-2}}+(-1)^{c_{n-3}+1}+(-1)^{c_{n-3}+c_{n-2}+1})\cdot
\widehat{f_0^{n-4}}(c^{n-4}) \\
&\ \ \ \ \ +(-1)^{c_{n-4}}((-1)+(-1)^{c_{n-2}})\cdot
\widehat{f_2^{n-4}}(c^{n-4})\\
&\ \ \ \ \ +(-1)^{c_{n-4}+c_{n-3}}((-1)+(-1)^{c_{n-2}})\cdot
\widehat{f_2^{n-4}}(c^{n-4}+e_{2^{n-5}})\\

&=\left\{
\begin{array}{ll}
-2(-1)^{c_{n-3}}\widehat{f_0^{n-4}}(c^{n-4})\ &\ if\ c_{n-2}=0, \\
-2\widehat{f_0^{n-4}}(c^{n-4})-4(-1)^{c_{n-4}}\widehat{f_0^{n-5}}(c^{n-5})\ &\ if\ c_{n-2}=1\ and\ c_{n-3}=0, \\
-2\widehat{f_0^{n-4}}(c^{n-4})-4(-1)^{c_{n-4}+c_{n-5}}\widehat{f_2^{n-5}}(c^{n-5}+e_{2^{n-6}})\ &\ if\ c_{n-2}=1\ and\ c_{n-3}=1.  %
\end{array}
\right.

\end{array}
\end{equation}

By (\ref{ff22-1}) and (\ref{ff22-2}), the relation for $f_2^n$
follows.

Similarly,
$\widehat{f_1^n}(c^n)=\widehat{f_1^{n-1}}(c^{n-1})+\sum\limits_{0\leq
j\leq 7}(-1)^{g_{1,j}^{n-4}}$, where $\sum\limits_{0\leq j\leq
7}(-1)^{g_{1,j}^{n-4}}$ can be calculated as
\begin{equation}
\begin{array}{ll}
&\sum\limits_{0\leq j\leq 7}(-1)^{g_{1,j}^{n-4}}\\
&=\left\{
\begin{array}{ll}
-2(-1)^{c_{n-3}}\widehat{f_1^{n-4}}(c^{n-4})\ &\ if\ c_{n-2}=1, \\
-2\widehat{f_1^{n-4}}(c^{n-4})-4(-1)^{c_{n-4}}\widehat{f_1^{n-5}}(c^{n-5})\ &\ if\ c_{n-2}=0\ and\ c_{n-3}=0, \\
-2\widehat{f_1^{n-4}}(c^{n-4})-4(-1)^{c_{n-4}+c_{n-5}}\widehat{f_3^{n-5}}(c^{n-5}+e_{1})\ &\ if\ c_{n-2}=0\ and\ c_{n-3}=1.  %
\end{array}
\right.
\end{array}
\end{equation}

Similarly again,
$\widehat{f_3^n}(c^n)=\widehat{f_1^{n-1}}(c^{n-1}+e_1)+\sum\limits_{0\leq
j\leq 7}(-1)^{g_{3,j}^{n-4}}$, where $\sum\limits_{0\leq j\leq
7}(-1)^{g_{3,j}^{n-4}}$ can be calculated as
\begin{equation}
\begin{array}{ll}
&\sum\limits_{0\leq j\leq 7}(-1)^{g_{3,j}^{n-4}}\\
&=\left\{
\begin{array}{ll}
2(-1)^{c_{n-3}}\widehat{f_1^{n-4}}(c^{n-4}+e_1)\ &\ if\ c_{n-2}=0, \\
2\widehat{f_1^{n-4}}(c^{n-4}+e_1)+4(-1)^{c_{n-4}}\widehat{f_1^{n-5}}(c^{n-5}+e_1)\ &\ if\ c_{n-2}=1\ and\ c_{n-3}=0, \\
2\widehat{f_1^{n-4}}(c^{n-4}+e_1)+4(-1)^{c_{n-4}+c_{n-5}}\widehat{f_3^{n-5}}(c^{n-5})\ &\ if\ c_{n-2}=1\ and\ c_{n-3}=1.  %
\end{array}
\right.
\end{array}
\end{equation}

\end{proof}

Cusick and St\u{a}nic\u{a}\cite{cusick} have proved that
$wt(F_3^n(x))=2(wt(F_3^{n-2}(x))+wt(F_3^{n-3}(x)))+2^{n-3}$, i.e.
$\widehat{F_3^{n}}(0)=2(\widehat{F_3^{n-2}}(0)+\widehat{F_3^{n-3}}(0))$
(in fact it could also be verified by Lemma \ref{high0} and Lemma
\ref{high1}). The following Lemma gives more relations about
$\widehat{F_3^{n}}(0)$.

\begin{lemma}\label{F0relation}
$\widehat{F_3^{n}}(0)$ satisfies the following relationships:
\begin{equation}
\begin{array}{lll}
\widehat{F_3^{n}}(0)&=\widehat{F_3^{n-1}}(0)+2\widehat{F_3^{n-4}}(0)+4\widehat{F_3^{n-5}}(0)
&n\geq 8,\\
\widehat{F_3^{n-1}}(0)&\leq \widehat{F_3^n}(0)\leq
2\widehat{F_3^{n-1}}(0),&n\geq 7.
\end{array}
\end{equation}
\end{lemma}
\begin{proof}
For the first equation, by the recurrence relation
$\widehat{F^{n}}(0)=2(\widehat{F^{n-2}}(0)+\widehat{F^{n-3}}(0))$
, we have for all $n\geq 8$,
\begin{equation}
\begin{array}{ll}
\widehat{F_3^{n}}(0)=2(\widehat{F_3^{n-2}}(0)+\widehat{F_3^{n-3}}(0)),\\
\widehat{F_3^{n-1}}(0)=2(\widehat{F_3^{n-3}}(0)+\widehat{F_3^{n-4}}(0)),\\
2\widehat{F_3^{n-2}}(0)=4(\widehat{F_3^{n-4}}(0)+\widehat{F_3^{n-5}}(0)),
\end{array}
\end{equation}
Calculating $"the\ first\ equation\ -\ the\ second\ equation\ +\
the\ third\ equation"$, we obtain
$$\widehat{F_3^{n}}(0)=\widehat{F_3^{n-1}}(0)+2\widehat{F_3^{n-4}}(0)+4\widehat{F_3^{n-5}}(0).$$

\begin{table}[thb]
\begin{center}
\begin{tabular}{|c|c|c|c|c|c|c|c|}\hline
$n=3$ & $n=4$ & $n=5$ & $n=6$ &$n=7$ & $n=8$ & $n=9$ & $n=10$\\
\hline
$6$ & $8$ & $20$ & $28$ &$56$ & $96$ & $168$ & $304$\\
\hline
\end{tabular}
\caption{ The values of $\widehat{F_3^n}(0)$.} \label{example-F30}
\end{center}
\end{table}

It is obvious $\widehat{F^{n-1}}(0)\leq \widehat{F^n}(0)$ for all
$n\geq 4$. For the proof of $\widehat{F^n}(0)\leq
2\widehat{F^{n-1}}(0)$, we show it by induction. From Table
\ref{example-F30}, it is true for $n<7$. Assume it is true for all
$n\leq s, n,s\geq 7$, we prove it for the case $s+1$. Since
\begin{equation}
\begin{array}{ll}
\widehat{F_3^{s-1}}(0)&\leq 2\widehat{F_3^{s-2}}(0),(by\ assumption)\\
\widehat{F_3^{s-2}}(0)&\leq 2\widehat{F_3^{s-3}}(0),(by\ assumption)\\
\widehat{F_3^{s}}(0)&=2(\widehat{F_3^{s-2}}(0)+\widehat{F_3^{s-3}}(0)),\\
\widehat{F_3^{s+1}}(0)&=2(\widehat{F_3^{s-1}}(0)+\widehat{F_3^{s-2}}(0)),
\end{array}
\end{equation}
It follows from the above relationships that
$$\widehat{F_3^{s+1}}(0)\leq 2\widehat{F_3^{s}}(0).$$
\end{proof}

\begin{lemma}\label{inequations}
Let $c^n=(c_0,\cdots,c_{n-1})\in \mathbb{F}_2^n$. If $c_{1}=1$,
then $$|\widehat{f_i^n}(c^n)|\leq \frac{1}{4}\cdot
\widehat{F_3^{n+2}}(0),(0\leq i\leq 3, n\geq 9).$$
\end{lemma}
\begin{proof}
We prove it by induction. Firstly with the help of computer, we
verify that for all $n\in [3,9],c^n\neq 0$,
$|\widehat{f_i^n}(c^n)|< \frac{1}{4}\cdot
\widehat{F_3^{n+2}}(0),(0\leq i\leq 3)$. (For example, see Table
\ref{example} for the case $n=6$. In this case
$\widehat{F_3^{n+2}}(0)=\widehat{F_3^{8}}(0)=96$, and we see that
$|\widehat{f_i^6}(c^6)|< \frac{1}{4}\cdot
\widehat{F_3^{8}}(0)=24,(0\leq i\leq 3)$). Assume the claim is
true for all $n<s$, where $n\geq 9, s\geq 10$, we now prove it is
true for $s$.

\begin{table}[thb]
\begin{center}
\begin{tabular}{|c|c|c|c|}\hline
$(0,36,28,28,4)$ & $(1,4,12,4,12)$ & $(2,12,20,4,-4)$ &
$(3,-4,-12,-4,4)$ \\ \hline $(4,12,-4,20,4)$ & $(5,-4,12,-4,-4)$ &
$(6,-12,4,-4,-4)$ & $(7,4,-12,4,4)$ \\ \hline $(8,12,20,-4,4)$ &
$(9,12,4,4,12)$ & $(10,4,-4,4,-4)$ & $(11,-12,-4,-4,4)$ \\ \hline
$(12,4,-12,-12,4)$ & $(13,-12,4,-4,-4)$ & $(14,-4,12,-4,-4)$ &
$(15,12,-4,4,4)$
\\ \hline
$(16,12,4,20,-4)$ & $(17,-4,4,-4,-12)$ & $(18,4,12,12,4)$ & $(19,4,-4,4,-4)$ \\
\hline $(20,4,4,-4,-4)$ & $(21,4,4,4,4)$ & $(22,-4,-4,-12,4)$ &
$(23,-4,-4,-4,-4)$ \\ \hline $(24,-12,-4,4,-4)$ &
$(25,4,-4,12,-12)$ & $(26,-4,-12,-4,4)$ & $(27,-4,4,-12,-4) $ \\
\hline $(28,-4,-4,12,-4)$ & $(29,-4,-4,-12,4)$ & $(30,4,4,4,4)$ &
$(31,4,4,12,-4)$ \\ \hline
$(32,4,4,12,12)$ & $(33,4,4,4,20)$ & $(34,-4,-4,4,-12)$ & $(35,-4,-4,-4,12)$ \\
\hline $(36,12,4,4,12)$ & $(37,-4,4,-4,4)$ & $(38,4,12,-4,-12)$ &
$(39,4,-4,4,12)$ \\ \hline $(40,-4,-4,12,-4)$ & $(41,-4,-4,4,4)$ &
$(42,4,4,4,4)$ & $(43,4,4,-4,-4)$ \\ \hline $(44,-12,-4,4,-4)$ &
$(45,4,-4,-4,-12)$ & $(46,-4,-12,-4,4)$ & $(47,-4,4,4,-4)$
\\ \hline
$(48,-4,-4,-12,4)$ & $(49,-4,-4,-4,12)$ & $(50,4,4,-4,-4)$ & $(51,4,4,4,20)$ \\
\hline $(52,-12,-4,-4,4)$ & $(53,4,-4,4,-4)$ & $(54,-4,-12,4,-4)$
& $(55,-4,4,-4,-12)$
\\ \hline
$(56,4,4,-12,4)$ & $(57,4,4,-4,12)$ & $(58,-4,-4,-4,-4)$ & $(59,-4,-4,4,-12)$ \\
 \hline
$(60,12,4,-4,4)$ & $(61,-4,4,4,-4)$ & $(62,4,12,4,-4)$ &
$(63,4,-4,-4,20)$ \\ \hline
\end{tabular}

\caption{$(c,\widehat{f_0^6}(c),\widehat{f_1^6}(c),\widehat{f_2^6}(c),\widehat{f_3^6}(c))$,
where $c=(c_0,\cdots,c_5)\in \mathbb{F}_2^6$ is represented by its
corresponding integer number $\sum\limits_{0\leq i\leq 5}c_i2^i$.}
\label{example}
\end{center}
\end{table}

Since $c_{1}=1$, we have
$c^n,c^{n-1},c^{n-2},c^{n-3},c^{n-4},c^{n-5}$  are all not zero
vectors.

If $c_{n-1}=0$, then by Lemma \ref{high0} and Lemma
\ref{F0relation}, we have
\begin{equation}
\begin{array}{ll}
&\left|\widehat{f_0^s}(c^s)\right|\\
&=\left|2(\widehat{f_0^{s-2}}(c^{s-2})+(-1)^{c_{n-2}}\cdot
\widehat{f_0^{s-3}}(c^{s-3}))\right| \\
&\leq \left|2(\widehat{f_0^{s-2}}(c^{s-2})\right|+2\left|
\widehat{f_0^{s-3}}(c^{s-3}))\right| \\
 &< \frac{1}{4}\cdot
(2(\widehat{F_3^{s}}(0)+\widehat{F_3^{s-1}}(0))) \\
&=\frac{1}{4}\cdot \widehat{F_3^{s+2}}(0).
\end{array}
\end{equation}

Similarly, the case for $|\widehat{f_i^n}(c^n)|< \frac{1}{4}\cdot
\widehat{F^{n+2}}(0),(i=1,2)$ can be proven.

For the case $i=3$, we have
\begin{equation}
\begin{array}{ll}
\left|\widehat{f_3^s}(c^s)\right|&=\left|2(-1)^{c_{s-2}}\cdot
\widehat{f_1^{s-3}}(c^{s-3}+e_1)\right| \\
&=\left|2\cdot \widehat{f_1^{s-3}}(c^{s-3}+e_1)\right|\\
& < \frac{1}{4}\cdot 2\widehat{F_3^{s-1}}(0) \\
&<\frac{1}{4}\cdot(
2\widehat{F_3^{s-1}}(0)+2\widehat{F_3^{s}}(0))\\
&= \frac{1}{4}\cdot\widehat{F_3^{s+2}}(0).
\end{array}
\end{equation}

If $c_{n-1}=1$, we prove the case $i=0,2$, and leave the proof for
the case $f_1^n,f_3^n$ to the reader since the recurrence forms
are similar. By Lemma \ref{high1} , for $i=0, 2$,
\begin{equation}
\begin{array}{ll}
\widehat{f_i^n}(c^n)\
&=\widehat{f_0^{n-1}}(c^{n-1})\pm2\cdot\widehat{f_0^{n-4}}(c^{n-4})), \\
\ \ \
or&=\widehat{f_0^{n-1}}(c^{n-1})\pm2\cdot\widehat{f_0^{n-4}}(c^{n-4}))\pm4\cdot\widehat{f_1^{n-5}}(c^{n-5}),
 \\
\ \ \
or&=\widehat{f_0^{n-1}}(c^{n-1})\pm2\cdot\widehat{f_0^{n-4}}(c^{n-4}))\pm4\cdot\widehat{f_1^{n-5}}(c^{n-5}+e_{2^{n-6}}).
\end{array}
\end{equation}

We prove the inequality for the first case and the second case,
while the third case is similar. If $\widehat{f_i^n}(c^n)
=\widehat{f_0^{n-1}}(c^{n-1})\pm2\cdot\widehat{f_0^{n-4}}(c^{n-4}))$,
then by Lemma \ref{F0relation} and induction,
\begin{equation}
\begin{array}{ll}
&\left|\widehat{f_i^s}(c^s)\right|\\
&\leq\left|\widehat{f_0^{s-1}}(c^{s-1})\right|+2\left|\widehat{f_0^{s-4}}(c^{s-4})\right|\\
&<\frac{1}{4}\cdot (\widehat{F_3^{s+1}}(0)+2\widehat{F_3^{s-2}}(0))\\
&<\frac{1}{4}\cdot (2\widehat{F_3^{s}}(0)+2\widehat{F_3^{s-1}}(0))\\
&=\frac{1}{4}\cdot \widehat{F_3^{s+2}}(0).
\end{array}
\end{equation}

When
$\widehat{f_i^n}(c^n)=\widehat{f_0^{n-1}}(c^{n-1})\pm2\cdot\widehat{f_0^{n-4}}(c^{n-4}))\pm4\cdot\widehat{f_1^{n-5}}(c^{n-5})
$, then by Lemma \ref{F0relation} and induction again,
\begin{equation}
\begin{array}{ll}
&\left|\widehat{f_i^s}(c^s)\right|\\
&<\frac{1}{4}\cdot (\widehat{F_3^{s+1}}(0)+2\widehat{F_3^{s-2}}(0)+4\widehat{F_3^{s-3}}(0))\\
&=\frac{1}{4}\cdot \widehat{F_3^{s+2}}(0).
\end{array}
\end{equation}

\end{proof}

\begin{theorem}
For all $c^n=(x_0,\cdots,x_{n-1})\neq 0$ and all $n\geq 3$,
$$\left|\widehat{F_3^n}(c^n)\right|< \widehat{F_3^{n}}(0).$$
\end{theorem}
\begin{proof}
For the few cases $n\leq 10$, we have the correctness by the
computer's computation results. Now assume $n>10$.

Since $c^n\neq 0$, by Lemma \ref{rots},
$\widehat{F_3^n}(x_0,\cdots,x_{n-1})=\widehat{F_3^n}(x_j,\cdots,x_{n-j-1})$
for all $j\in [0,n-1]$. Thus we assume $c_{1}=1$. By Lemma
\ref{inequations}, we have
\begin{equation*}
\begin{array}{ll}
&\left|\widehat{F_3^n}(c^n)\right|\\
&=\left|\widehat{f_0^{n-2}}(c^{n-2})+(-1)^{c_{n-2}}\cdot
\widehat{f_2^{n-2}}(c^{n-2}) +(-1)^{c_{n-1}}\cdot
\widehat{f_1^{n-2}}(c^{n-2})+(-1)^{c_{n-2}+c_{n-1}}\cdot
\widehat{f_3^{n-2}}(c^{n-2})\right| \\

&\leq\left|\widehat{f_0^{n-2}}(c^{n-2})\right|+\left|
\widehat{f_2^{n-2}}(c^{n-2})\right| +\left|
\widehat{f_1^{n-2}}(c^{n-2})\right|+\left|
\widehat{f_3^{n-2}}(c^{n-2})\right| \\

&<\frac{1}{4}\cdot (\widehat{F_3^{n}}(0)+\widehat{F_3^{n}}(0)+\widehat{F_3^{n}}(0)+\widehat{F_3^{n}}(0))\\
&=\widehat{F_3^{n}}(0).
\end{array}
\end{equation*}

\end{proof}

\section{Conclusion}
 In this paper we prove the conjecture proposed in \cite{cusick},
 i.e. the nonlinearity of $F_3^n(x_0,\cdots,x_{n-1})$ is the same as its
 weight.  Recently Cusick remarked that computer's results imply that the Conjecture
may be extended to RSBF with SANF $x_0x_ax_b(b>a>0)$ in the case
of odd $n$. However it seems difficult to prove that. It is
interesting to note that it has been proved in \cite{kim} that the
nonlinearity of $F_2^n(x_0,\cdots,x_{n-1})=\sum\limits_{0\leq
i\leq n-1}x_ix_{i+s(mod\ n)}$ is the same as its weight if
$\frac{n}{gcd(n,s)}$ is even. These properties show that rotation
symmetric Boolean functions have nice cryptographic applications.
Whether higher degree RSBFs have these properties is an
interesting topic for further research.

\end{document}